\documentclass[letterpaper]{article}
\usepackage{aaai}
\usepackage{times}
\usepackage[scaled]{helvet}
\usepackage{courier}
\usepackage{balance}

\usepackage{amsmath}
\usepackage{amsthm}
\usepackage{amsfonts}
\usepackage{enumerate}
\usepackage{thm-restate}
\usepackage{soul}
\usepackage{multirow}
\usepackage[draft]{ifdraft}
\usepackage{macros}
\usepackage{url}
\allowdisplaybreaks

\title{Handling \owl{sameAs} via Rewriting\vspace*{-2cm}}
\author{
}

\newcommand{\TRurl}{\url{http://tinyurl.com/k9clzk6}}

\begin{document}
\maketitle

\begin{abstract}
Rewriting is widely used to optimise \sameAs reasoning in materialisation based
OWL 2 RL systems. We investigate issues related to both the correctness and
efficiency of rewriting, and present an algorithm that guarantees correctness,
improves efficiency, and can be effectively parallelised. Our evaluation shows
that our approach can reduce reasoning times on practical data sets by orders
of magnitude.
\end{abstract}

\section{Introduction}

RDF \cite{manola2004rdf} and SPARQL \cite{sparql11-overview-w3c} are
increasingly being used to store and access semistructured data. An OWL
ontology \cite{motik2009owl2} is often used to enhance query answers with
tuples implied by the ontology and data, and the OWL 2 RL profile was
specifically designed to allow for tractable rule-based query
answering~\cite{owl2-profiles}. In practice, this often involves using a
forward chaining procedure in which the \emph{materialisation} (i.e., all
consequences) of the ontology and data is computed in a preprocessing step,
allowing queries to be evaluated directly over the materialised triples. This
technique is used by systems such as Owlgres \cite{DBLP:conf/owled/StockerS08},
WebPIE \cite{DBLP:journals/ws/UrbaniKMHB12}, Oracle's RDF store
\cite{DBLP:conf/icde/WuEDCKAS08}, OWLIM SE
\cite{DBLP:journals/semweb/BishopKOPTV11}, and RDFox \cite{MNPHO14a}.

One disadvantage of materialisation is that the preprocessing step can be
costly w.r.t.\ both the computation and storage of entailed triples. This
problem is exacerbated when materialisation requires equality reasoning---that
is, when the \sameAs property is used to state equalities between resources.
OWL 2 RL/RDF \cite[Section 4.3]{owl2-profiles} axiomatises the semantics of
\sameAs using rules such as ${\triple{s'}{p}{o} \leftarrow \triple{s}{p}{o}
\wedge \triple{s}{\sameAs}{s'}}$ that, for each pair of equal resources $r$ and
$r'$, `copy' all triples between $r$ and $r'$. It is well known that such
`copying' can severely impact both the materialisation size and time
\cite{DBLP:conf/semweb/KolovskiWE10}; what is less obvious is that the increase
in computation time due to duplicate derivations may be even more serious (see
Section~\ref{sec:motivation}).

In order to address this problem, materialisation based systems often use some
form of \emph{rewriting}---a well-known technique for theorem proving with
equality \cite{baader98term,NieuwenhuisRubio:HandbookAR:paramodulation:2001}.
In the OWL 2 RL setting, rewriting consists of choosing one representative from
each set of equal resources, and replacing all remaining resources in the set
with the representative. Variants of this idea have been implemented in many of
the above mentioned systems, and they have been shown to be very effective on
practical data sets \cite{DBLP:conf/semweb/KolovskiWE10}.

Although the idea of rewriting is well known, ensuring its correctness (i.e.,
ensuring that the answer to an arbitrary SPARQL query is the same with and
without rewriting) is not straightforward. In this paper we identify two
problems that, we believe, have been commonly overlooked in existing
implementations. First, whenever a resource $r$ is rewritten in the data, $r$
must also be rewritten in the rules; hence, the rule set cannot be assumed to
be fixed during the course of materialisation, which is particularly
problematic if computation is paralellised. Second, it is a common assumption
that SPARQL queries can be efficiently evaluated over the materialisation by
rewriting them, evaluating them over the rewritten triples, and then
`expanding' the answer set (i.e., substituting all representative resources
with equal ones in all possible ways). However, such an approach can be
incorrect when SPARQL queries are evaluated under bag semantics, or when they
contain builtin functions. 

We address both issues in this paper and make the following contributions. In
Section~\ref{sec:motivation} we discuss the problems related to \sameAs in more
detail and show how they can lead to both increased computation costs and
incorrect query answers. In Section~\ref{sec:reasoning} we present an algorithm
that generalises OWL 2 RL materialisation, can also handle SWRL rules
\cite{swrl-w3c}, rewrites rules as well as data triples, and is
\emph{lock-free} \cite{DBLP:books/daglib/0020056}. The latter means that at
least one thread always makes progress, ensuring that the system is less
susceptible to adverse thread scheduling decisions and thus scales better to
many threads. In Section~\ref{sec:querying} we show how to modify SPARQL query
processing so as to guarantee correctness. Finally, in
Section~\ref{sec:evaluation} we present a preliminary evaluation of an
implementation of our algorithms based on the open-source RDFox system. We show
that rewriting can reduce the number of materialised triples by a factor of up
to 7.8, and can reduce materialisation time by a factor of up to 31.1 on a
single thread, with the time saving being largely due to the elimination of
duplicate derivations. Our approach also parallelises computation very well,
providing a speedup of up to 6.7 with eight physical cores, and up to 9.6 with
16 virtual cores.\footnote{In \emph{hyperthreading}, two virtual cores have
their own architectural state, but share the execution resources of one
physical core.} Note, that datalog resoning 
is PTIME complete in the size of the data and is thus deemed to be inherently sequential.

Due to space considerations, in this paper we have only been able to present a
high level description of our algorithms, \ifdraft{but detailed formalisations
and correctness proofs are provided in the appendix, and the implemented system
and all test data sets are available online.\footnote{\TRurl}}{but the paper is
complemented by an online technical report that includes formalisations and
correctness proofs for our algorithms;\footnote{\TRurl} the implemented system
and all test data sets are also available at the same location.}
\section{Preliminaries}\label{sec:preliminaries}

\noindent\textbf{OWL 2 RL and RDF.}\quad
A \emph{term} is a \emph{resource} (i.e., a constant) or a variable. Unless
otherwise stated, $s$, $p$, $o$, and $t$ are terms, and $x$, $y$, and $z$ are
variables. An \emph{atom} is a triple of terms $\triple{s}{p}{o}$ called the
\emph{subject}, \emph{predicate}, and \emph{object}, respectively. A
\emph{fact} (or \emph{triple}) is a variable-free atom. A \emph{rule} $r$ is an
implication of the form \eqref{eq:rule-form}, where ${\head{r} \defeq
\triple{s}{p}{o}}$ is the \emph{head}, ${\body{r} \defeq \triple{s_1}{p_1}{o_1}
\land\ldots\land \triple{s_n}{p_n}{o_n}}$ is the \emph{body}, and each variable
in $\head{r}$ also occurs in $\body{r}$.
\begin{align}
	\label{eq:rule-form} \triple{s}{p}{o} \leftarrow \triple{s_1}{p_1}{o_1} \land \ldots \land \triple{s_n}{p_n}{o_n}
\end{align}
A \emph{program} $P$ is a finite set of rules, and $P^\infty(\efacts)$ is the
\emph{materialisation} of $P$ on a finite set of \emph{explicit} (i.e.,
extensional or EDB) facts $\efacts$ \cite{abiteboul95foundation}.

Two styles of OWL 2 RL reasoning are known, corresponding to the RDF- and
DL-style semantics of OWL. In the RDF style, an ontology is represented using
triples stored with the data in a single RDF graph, and a fixed (i.e.,
independent from the ontology) set of rules is used to axiomatise the RDF-style
semantics \cite[Section 4.3]{owl2-profiles}. While conceptually simple, this
approach is inefficient because the fixed program contains complex joins. In
the DL style, the rules are derived from and depend on the ontology
\cite{GHVD03}, but they are shorter and contain fewer joins. This approach is
complete only if the ontology and the data satisfy conditions from Section 3 of
\cite{motik2009owl2}---an assumption commonly met in practice. Rewriting can be
used with either style of reasoning, but we will use the DL style in our
examples and evaluation because the rules are more readable and their
evaluation tends to be more efficient.
\section{Problems with \sameAs}\label{sec:motivation}

In this section we discuss, by means of an example, the problems that the
\sameAs property poses to materialisation-based reasoners. The semantics of
\sameAs can be captured explicitly using program $\Peq$, consisting of rules
\eqref{eq:eq1}--\eqref{eq:eq5}, which axiomatises \sameAs as a congruence
relation (i.e., an equivalence relation satisfying the replacement property).
We call each set of resources all of which are equal to each other an
\sameAs-\emph{clique}.
\begin{align}
    & \triple{x_i}{\sameAs}{x_i}  \leftarrow \triple{x_1}{x_2}{x_3}\text{, for }1\leq i\leq 3       \tag{$\approx_1$}\label{eq:eq1}\\
    & \triple{x_1'}{x_2}{x_3}     \leftarrow \triple{x_1}{x_2}{x_3}\land\triple{x_1}{\sameAs}{x_1'} \tag{$\approx_2$}\label{eq:eq2}\\
    & \triple{x_1}{x_2'}{x_3}     \leftarrow \triple{x_1}{x_2}{x_3}\land\triple{x_2}{\sameAs}{x_2'} \tag{$\approx_3$}\label{eq:eq3}\\
    & \triple{x_1}{x_2}{x_3'}     \leftarrow \triple{x_1}{x_2}{x_3}\land\triple{x_3}{\sameAs}{x_3'} \tag{$\approx_4$}\label{eq:eq4}\\
    & \false                      \leftarrow \triple{x}{\differentFrom}{x}                          \tag{$\approx_5$}\label{eq:eq5}
\end{align}
OWL 2 RL/RDF \cite[Section 4.3]{owl2-profiles} also makes \sameAs symmetric and
transitive, but those rules are redundant as they are instances of
\eqref{eq:eq2} and \eqref{eq:eq4}.

Rules \eqref{eq:eq1}--\eqref{eq:eq5} can lead to the derivation of many
equivalent triples, as we demonstrate using an example program $\Pex$
containing rules \eqref{eq:R}--\eqref{eq:F3}; these correspond directly to SWRL
rules, but one could equally use slightly more complex rules obtained from OWL
2 RL axioms.
\begin{align}
    & \triple{x}{\sameAs}{\duri{USA}} \leftarrow\triple{\duri{Obama}}{\duri{presidentOf}}{x} \tag{$R$}\label{eq:R}\\
    & \triple{x}{\sameAs}{\duri{Obama}} \leftarrow\triple{x}{\duri{presidentOf}}{\duri{USA}} \tag{$S$}\label{eq:S}\\
    & \triple{\duri{USPresident}}{\duri{presidentOf}}{\duri{US}}                             \tag{$F_1$}\label{eq:F1}\\
    & \triple{\duri{Obama}}{\duri{presidentOf}}{\duri{America}}                              \tag{$F_2$}\label{eq:F2}\\
    & \triple{\duri{Obama}}{\duri{presidentOf}}{\duri{US}}                                   \tag{$F_3$}\label{eq:F3}
\end{align}
On ${\Pex \cup \Peq}$, rule \eqref{eq:R} derives that \duri{USA} is equal to
\duri{US} and \duri{America}, and rules \eqref{eq:eq1}--\eqref{eq:eq4} then
derive an \sameAs triple for each of the nine pairs involving \duri{USA},
\duri{America}, and \duri{US}. The total number of derivations, however, is
much higher: we derive each triple once from rule \eqref{eq:eq1}, three times
from rule \eqref{eq:eq2}, once from rule \eqref{eq:eq3},\footnote{Rule
\eqref{eq:eq1} derives $\triple{\sameAs}{\sameAs}{\sameAs}$, so we can map
variable $x_2$ to \sameAs in rule \eqref{eq:eq3}.} and three times from rule
\eqref{eq:eq4}; thus, we get 66 derivations in total for the nine \sameAs
triples. Analogously, rule \eqref{eq:S} derives that \duri{Obama} and
\duri{USPresident} are equal, and rules \eqref{eq:eq1}--\eqref{eq:eq4} derive
the two \sameAs triples 22 times in total. These \sameAs triples lead to
further inferences; for example, from \eqref{eq:F1}, rules \eqref{eq:eq2} and
\eqref{eq:eq4} infer $2\times 3$ triples with subject \duri{Obama} or
\duri{USPresident}, and object \duri{USA}, \duri{America}, or \duri{US}. Each
of these six triples is inferred three times from rule \eqref{eq:eq2}, once
from rule \eqref{eq:eq3}, and three times from rule \eqref{eq:eq4}, so we get
36 derivations in total.

Thus, for each \sameAs-clique of size $n$, rules \eqref{eq:eq1}--\eqref{eq:eq4}
derive $n^2$ \sameAs triples via $2n^3+n^2+n$ derivations. Moreover, each
triple $\triple{s}{p}{o}$ with terms in \sameAs-cliques of sizes $n_s$, $n_p$,
and $n_o$, respectively, is `expanded' to $n_s\times n_p\times n_o$ triples,
each of which is derived ${n_s+n_p+n_o}$ times. This duplication of facts and
derivations is a major source of inefficiency.

To reduce these numbers, we can choose a representative resource for each
\sameAs-clique and then \emph{rewrite} all triples---that is, replace all
resources with their representatives
\cite{DBLP:conf/owled/StockerS08,DBLP:journals/ws/UrbaniKMHB12,DBLP:conf/semweb/KolovskiWE10,DBLP:journals/semweb/BishopKOPTV11}.
For example, after applying rule \eqref{eq:R}, we can choose \duri{USA} as the
representative of \duri{USA}, \duri{US} and \duri{America}, and, after applying
rule \eqref{eq:S}, we can choose \duri{Obama} as the representative of
\duri{Obama} and \duri{USPresident}. The materialisation of $\Pex$ then
contains only the triple $\triple{\duri{Obama}}{\duri{presidentOf}}{\duri{US}}$
and, as we show in Section~\ref{sec:reasoning}, the number of derivations of
\sameAs triples drops from over 60 to just 6.

Since \sameAs triples can be derived continuously during materialisation,
rewriting cannot be applied as preprocessing; moreover, to ensure that
rewriting does not affect query answers, the resulting materialisation must be
equivalent, modulo rewriting, to $[\Pex \cup \Peq]^\infty(\efacts)$. Thus, we
may need to continuously rewrite both triples and rules: rewriting only triples
can be insufficient. For example, if we choose \duri{US} as the representative
of \duri{USA}, \duri{US} and \duri{America}, then rule \eqref{eq:S} will not be
applicable, and we will fail to derive that \duri{USPresident} is equal to
\duri{Obama}. To the best of our knowledge, no existing system implements rule
rewriting; certainly OWLIM SE and Oracle's RDF store do not,\footnote{Personal
communication.} and so rewriting is \emph{not} guaranteed to preserve query
answers.

Note that the problem is less acute when using a fixed rule set operating on
(the triple encoding of) the ontology and data, but it can still arise if
\sameAs triples involve $\rdf{}$ or $\owl{}$ resources (with a fixed rule set,
these are the only resources occurring in rule bodies).
\section{Parallel Reasoning With Rewriting}\label{sec:reasoning}

The algorithm by \citeA{MNPHO14a} used in the RDFox system implements a
fact-at-a-time version of the semina{\"i}ve algorithm
\cite{abiteboul95foundation}: it initialises the set of facts $\tfacts$ with
the input data $\efacts$, and then computes $P^\infty(\efacts)$ by repeatedly
applying rules from $P$ to $\tfacts$ using $N$ threads until no new facts are
derived. The objective of our approach is to adapt the RDFox algorithm to use
rewriting and thus reduce both the size of $\tfacts$ and the time required to
compute it, while ensuring that an arbitrary SPARQL query can be answered over
the resulting facts as if the query were evaluated directly over $[P \cup
\Peq]^\infty(E)$. To achieve this, we use a mapping $\rho$ that maps resources
to their representatives. For $\alpha$ a fact, a rule, or a set thereof,
$\rho(\alpha)$ is obtained by replacing each resource $r$ in $\alpha$ with
$\rho(r)$; moreover, ${\expand{\tfacts}{\rho} \defeq \{ \triple{s}{p}{o} \mid
\triple{\rho(s)}{\rho(p)}{\rho(o)} \in \tfacts \}}$ is the \emph{expansion} of
$\tfacts$ with $\rho$. To promote concurrency, we update $\rho$ in a lock-free
way, using \emph{compare-and-set} primitives to prevent thread interference.
Moreover, we do not lock $\rho$ when computing $\rho(\alpha)$; instead, we only
require $\rho(\alpha)$ to be at least as current as $\alpha$ just before the
computation. For example, if $\rho$ is the identity as we start computing
${\rho(\triple{a}{b}{a})}$, and another thread makes $a'$ the representative of
$a$, then $\triple{a}{b}{a}$, $\triple{a'}{b}{a}$, $\triple{a}{b}{a'}$, and
$\triple{a'}{b}{a'}$ are all valid results.

We also maintain queues $R$ and $C$ of rewritten rules and resources,
respectively, for which also use lock-free implementations as described by
\citeA{DBLP:books/daglib/0020056}.

To extend the original RDFox algorithm with rewriting, we allow each thread to
perform three different actions. First, a thread can extract a rule $r$ from
the queue $R$ of rewritten rules and apply $r$ to the set of all facts
$\tfacts$, thus ensuring that changes to resources in rules are taken into
account.

Second, a thread can rewrite outdated facts---that is, facts containing a
resource that is not a representative of itself. To avoid iteration over all
facts in $\tfacts$, the thread extracts a resource $c$ from the queue $C$ of
unprocessed outdated resources, and uses indexes by \citeA{MNPHO14a} to
identify each fact ${F \in \tfacts}$ containing $c$. The thread then removes
each such $F$ from $\tfacts$, and it adds $\rho(F)$ to $\tfacts$.

Third, a thread can extract and process an unprocessed fact $F$ in $\tfacts$.
The thread first checks whether $F$ is outdated (i.e., whether ${F \neq
\rho(F)}$); if so, the thread removes $F$ from $\tfacts$ and adds $\rho(F)$ to
$\tfacts$. If $F$ is not outdated but is of the form $\triple{a}{\sameAs}{b}$
with ${a \neq b}$, the thread chooses a representative of the two resources,
updates $\rho$, and adds the other resource to queue $C$. The thread derives a
contradiction if $F$ is of the form $\triple{a}{\differentFrom}{a}$. Otherwise,
the thread processes $F$ by partially instantiating the rules in $P$ containing
a body atom that matches $F$, and applying such rules to $\tfacts$ as described
by \citeA{MNPHO14a}.

\medskip

Rewriting rules is nontrivial: RDFox uses an index to efficiently identify
rules matching a fact, and the index may need updating when $\rho$ changes.
Updating the index in parallel would be very complex, so we perform this
operation serially: when all threads are waiting (i.e., when all facts have
been processed), a single thread updates $P$ to $\rho(P)$, reindexes it, and
inserts the updated rules (if any) into the queue $R$ of rules for
reevaluation. This is obviously a paralellisation bottleneck, but our
experiments have shown that the time used for this process is not significant
when programs are of moderate size.

Parallel modification of $\tfacts$ can also be problematic, as the following
example demonstrates: (1)~thread A extracts a current fact $F$; (2)~thread B
updates $\rho$ and deletes an outdated fact $F'$; and (3)~thread A derives $F'$
from $F$ and writes $F'$ into $\tfacts$, thus undoing the work of thread B.
This could be solved via locking, but at the expense of parallelisation. Thus,
instead of physically removing facts from $\tfacts$, we just mark them as
outdated; then, when matching the body atoms of partially instantiated rules,
we simply skip all marked facts. All this can be done lock-free, and we can
remove all marked facts in a postprocessing step.

\medskip

Theorem~\ref{thm:correctness} states several important properties of our
algorithm that, taken together, ensure the algorithm's correctness; a detailed
formalisation of the algorithm and a proof of the theorem are \ifdraft{given in
the appendix.}{provided in the online technical report.}

\begin{restatable}{theorem}{thmcorrectness}\label{thm:correctness}
    The algorithm terminates for each finite set of facts $\efacts$ and program
    $P$. Let $\rho$ be the final mapping and let $\tfacts$ be the final set of
    unmarked facts.
    \begin{enumerate}
        \item $\triple{a}{\sameAs}{b} \in \tfacts$ implies ${a = b}$---that is,
        $\rho$ captures all equalities.

        \item ${F \in \tfacts}$ implies ${\rho(F) = F}$---that is, $\tfacts$ is
        minimal.

        \item ${\expand{\tfacts}{\rho} = [P \cup
        \Peq]^{\infty}(\efacts)}$---that is, $\tfacts$ and $\rho$ together
        represent $[P \cup \Peq]^{\infty}(\efacts)$.
    \end{enumerate}
\end{restatable}

\subsection{Example}

Table~\ref{tab:run} shows six steps of an application of our algorithm to the
example program $\Pex$ from Section \ref{sec:motivation} on one thread. Some
resource names have been abbreviated for convenience, and $\approx$ abbreviates
\sameAs. The $\triangleright$ symbol identifies the last fact extracted from
$\tfacts$. Facts are numbered for easier referencing, and their (re)derivation
is indicated on the right: $R(n)$ or $S(n)$ means that the fact was obtained
from fact $n$ and rule $R$ or $S$; moreover, we rewrite facts immediately after
merging resources, so $W(n)$ identifies a rewritten version of fact $n$, and
$M(n)$ means that a fact was marked outdated because fact $n$ caused $\rho$ to
change.

We start by extracting facts from $\tfacts$ and, in steps 1 and 2, we apply rule
$R$ to facts 2 and 3 to derive facts 4 and 5, respectively. In step 3, we
extract fact 4, merge \duri{America} into \duri{USA}, mark facts 2 and 4 as
outdated, and add their rewriting, facts 6 and 7, to $\tfacts$. In step 4 we
merge $\duri{USA}$ into $\duri{US}$, after which there are no further facts to
process. Mapping $\rho$, however, has changed, so we update $P$ to contain rules
\eqref{eq:R1} and \eqref{eq:S1}, and add them to the queue $R$.
\begin{align}
    & \triple{x}{\sameAs}{\duri{US}}    \leftarrow \triple{\duri{Obama}}{\duri{presidentOf}}{x} \tag{$R'$}\label{eq:R1}\\
    & \triple{x}{\sameAs}{\duri{Obama}} \leftarrow \triple{x}{\duri{presidentOf}}{\duri{US}}    \tag{$S'$}\label{eq:S1}
\end{align}
In step 5 we evaluate the rules in queue $R$, which introduces facts 9 and 10.
Finally, in step 6, we rewrite \duri{USPresident} into \duri{Obama} and mark
facts 1 and 9 as outdated. At this point the algorithm terminates, making only
six derivations in total, instead of more than 60 derivations when \sameAs is
axiomatised explicitly (see Section \ref{sec:motivation}).

\begin{table*}[tb]
\caption{An Example Run of \refalg{mat} on $\Pex$ and One Thread}\label{tab:run}
\newcommand{\shortSameAs}{{\approx}}
\newcommand{\Fi}{$\triple{\duri{USPres}}{\duri{presOf}}{\duri{US}}$}
\newcommand{\Fii}{$\triple{\duri{Obama}}{\duri{presOf}}{\duri{Am}}$}
\newcommand{\Fiii}{$\triple{\duri{Obama}}{\duri{presOf}}{\duri{US}}$}
\newcommand{\Fiv}{$\triple{\duri{Am}}{\shortSameAs}{\duri{USA}}$}
\newcommand{\Fv}{$\triple{\duri{US}}{\shortSameAs}{\duri{USA}}$}
\newcommand{\Fvi}{$\triple{\duri{Obama}}{\duri{presOf}}{\duri{USA}}$}
\newcommand{\Fvii}{$\triple{\duri{USA}}{\shortSameAs}{\duri{USA}}$}
\newcommand{\Fviii}{$\triple{\duri{US}}{\shortSameAs}{\duri{US}}$}
\newcommand{\Fix}{$\triple{\duri{USPres}}{\shortSameAs}{\duri{Obama}}$}
\newcommand{\Fx}{$\triple{\duri{Obama}}{\shortSameAs}{\duri{Obama}}$}
\newcommand{\tarrow}{$\triangleright$}
\centering\scriptsize{
\begin{tabular}{rl@{\hspace{.1cm}}l||rl@{}l||rll}
\hline
\multicolumn{3}{c}{\textbf{Step 1}}     & \multicolumn{3}{c}{\textbf{Step 2}}   & \multicolumn{3}{c}{\textbf{Step 3}}     \\
\hline
1           & \Fi           &           & 1             & \Fi     &             & 1             & \Fi       &             \\
\tarrow 2   & \Fii          &           & 2             & \Fii    &             & 2             & \st{\Fii} & $M(4)$      \\
3           & \Fiii         &           & \tarrow 3     & \Fiii   &             & 3             & \Fiii     &             \\
\cline{1-3}
4           & \Fiv          & $R(2)$    & 4             & \Fiv    &             & \tarrow 4     & \st{\Fiv} & $M(4)$      \\
                                        \cline{4-6}
            &               &           & 5             & \Fv     & $R(3)$      & 5             & \Fv       &             \\
                                                                                \cline{7-9}
            &               &           &               &         &             & 6             & \Fvi      & $W(2)$      \\
            &               &           &               &         &             & 7             & \Fvii     & $W(4)$      \\
\multicolumn{9}{c}{}                                                                                                      \\
\hline
\multicolumn{3}{c}{\textbf{Step 4}}     & \multicolumn{3}{c}{\textbf{Step 5}}   & \multicolumn{3}{c}{\textbf{Step 6}}     \\
\hline
1           & \Fi           &           & 1             & \Fi     &             & 1             & \st{\Fi}  & $M(9)$      \\
3           & \Fiii         & $W(6)$    & 3             & \Fiii   &             & 3             & \Fiii     &             \\
\tarrow 5   & \st{\Fv}      & $M(5)$    & 8             & \Fviii  & $R'(3)$     & 8             & \Fviii    &             \\
                                        \cline{4-6}
6           & \st{\Fvi}     & $M(5)$    & 9             & \Fix    & $S'(1)$     & \tarrow 9     & \st{\Fix} & $M(9)$      \\ 
7           & \st{\Fvii}    & $M(5)$    & 10            & \Fx     & $S'(3)$     & 10            & \Fx       & $W(9)$      \\
\cline{1-3}                                                                     \cline{7-9}
8           & \Fviii        & $W(5,7)$  &               &         &             &               &           &             \\
\multicolumn{9}{c}{}                                                                                                      \\
\end{tabular}}
\end{table*}
\section{SPARQL Queries on Rewritten Triples}\label{sec:querying}
Given a set of facts $\tfacts$ and mapping $\rho$, the 
expected answers to a SPARQL query $Q$ are those obtained by evaluating 
$Q$ in the expansion $\expand{\tfacts}{\rho}$. Evaluating $Q$ on 
$\expand{\tfacts}{\rho}$, however, forgoes any advantage of smaller 
joins obtained from evaluating $Q$ on the succinct representation 
$\tfacts$.
Thus, the question arises how $Q$ can be evaluated on $\tfacts$ yielding 
the answers in $\expand{\tfacts}{\rho}$ whilst only necessary resources are expanded. To illustrate our strategy, we
use our program $\Pex$ from Section \ref{sec:motivation}: Recall that, after we finish
the materialisation of $\Pex$, we have ${\rho(x) = \duri{US}}$ for each ${x \in
\{ \duri{USA}, \duri{AM}, \duri{US} \}}$ and ${\rho(x) = \duri{Obama}}$ for
each ${x \in \{ \duri{USPresident},\duri{Obama} \}}$.

Firstly, we discuss query evaluation under SPARQL bag semantics where repeated
answers matter. To this end, let
\begin{align*}\scriptstyle
   Q_1 \ := \ \mathsf{SELECT} \; ?x \; \mathsf{WHERE} \; \{ \; ?x \; \duri{presidentOf} \; ?y \; \} 
\end{align*}
On $\expand{T}{\rho}$, query $Q_1$ produces answers ${\mu_1 = \{ ?x \mapsto
\duri{Obama} \}}$ and ${\mu_2 = \{ ?x \mapsto \duri{USPresident} \}}$, each of
which is repeated three times---once for each match of $?y$ to \duri{USA},
\duri{US}, or \duri{America}.
A na\"ive evaluation of the normalised query $\rho(Q_1)$ on $\tfacts$ coupled 
with a post-hoc expansion under $\rho$ produces one occurrence of each
$\mu_1$ and $\mu_2$ which is not the intended result; 
This problem arises because the final expansion step does not take into account
the number of times each binding of $?y$ contributes to the result. We therefore
modify the projection operator to output each projected answer as many times as
there are resources in the projected \sameAs-clique(s). Thus, we answer
$Q_1$ as follows: we match the triple pattern of $\rho(Q_1)$ to $\tfacts$ as
usual, obtaining one answer ${\nu_1 = \{ ?x \mapsto \duri{Obama}, ?y \mapsto
\duri{US} \}}$; then, we project $?y$ from $\nu_1$ and obtain three occurrences
of $\mu_1$ since the \sameAs-clique of $\duri{US}$ is of size three; finally,
we expand each occurrence of $\mu_1$ to $\mu_2$ to obtain all six results.

Secondly, we treat query evaluation in the presence of SPARQL builtin 
functions. Let $Q_2$ be as follows:
\begin{align*}\scriptstyle
    \mathsf{SELECT} \; ?y \; \mathsf{WHERE} \; \{ \; ?x \; \duri{presidentOf} \; \duri{US} \; . \; \mathsf{BIND}(\mathsf{STR}(?x) \; \mathsf{AS} \; ?y) \; \}
\end{align*}
On $\expand{T}{\rho}$, query $Q_2$ produces answers ${\tau_1 = \{ ?y \mapsto
\strlit{Obama} \}}$ and ${\tau_2 = \{ ?y \mapsto \strlit{USPresident} \}}$; in
contrast, on $\tfacts$, query $\rho(Q_2)$ yields only $\tau_1$, which does not
expand into $\tau_2$ because the strings ``Obama'' and ``USPresident'' are not
equal. Our evaluation therefore expands answers \emph{before} evaluating
builtin functions. Thus, we answer $Q_2$ as follows: we match the triple
pattern of $\rho(Q_2)$ to $\tfacts$ as usual, obtaining ${\kappa_1 = \{ ?x
\mapsto \duri{Obama} \}}$; then, we expand $\kappa_1$ to ${\kappa_2 = \{ ?x
\mapsto \duri{USPresident} \}}$; next, we evaluate the $\mathsf{BIND}$
expression and extend $\kappa_1$ and $\kappa_2$ with the respective values for
$?y$; finally, we project $?x$ to obtain $\tau_1$ and $\tau_2$. Since we have
already expanded $?x$, we must not repeat the projected answers further; instead, we output each
projected answer only once to obtain the correct answer cardinalities.
\section{Evaluation}\label{sec:evaluation}

We have implemented our approach as an extension to RDFox, allowing the system
to handle \sameAs via rewriting (REW) or the axiomatisation (AX) from Section
\ref{sec:motivation}. We then compared the performance of materialisation using
these two approaches. In particular, we investigated the scalability of each
approach with the number of threads, and we measured the effect that rewriting
has on the number of derivations and materialised triples.

\medskip

\noindent\textbf{Test Data Sets.}\quad
We used five test data sets, each consisting of an OWL 2 DL ontology and a set
of facts. The data sets were chosen because they contain axioms with the
\sameAs property leading to interesting inferences. Four data sets were derived
from real-world applications.
\begin{itemize}
    \item Claros has been developed in an international collaboration between
    IT experts and archaeology and classical art research institutions with the
    aim of integrating disparate cultural heritage
    databases.\footnote{\url{http://www.clarosnet.org/XDB/ASP/clarosHome/}}

    \item DBpedia is a crowd-sourced community effort to extract structured
    information from Wikipedia and make this information available on the
    Web.\footnote{\url{http://www.dbpedia.org/}}

    \item OpenCyc is an extensive ontology about general human knowledge. It
    contains hundreds of thousands of terms organised in a carefully designed
    ontology and can be used as the basis of a wide variety of intelligent
    applications.\footnote{\url{http://www.cyc.com/platform/opencyc/}}

    \item UniProt is a subset of an extensive knowledge base about protein
    sequences and functional
    information.\footnote{\url{http://www.uniprot.org/}}
\end{itemize}
The ontologies of all data sets other than DBpedia are not in the OWL 2 RL
profile, so we first discarded all axioms outside OWL 2 RL, and then we
translated the remaining axioms into rules as described in \cite{GHVD03}.

Our fifth data set was UOBM \cite{DBLP:conf/esws/MaYQXPL06}---a synthetic data
set that extends the well-known LUBM \cite{DBLP:journals/ws/GuoPH05} benchmark.
We did not use LUBM because neither its ontology nor its data uses the \sameAs
property. The UOBM ontology is also outside OWL 2 RL; however, instead of using
its OWL 2 RL subset, we used its \emph{upper bound}
\cite{DBLP:conf/www/ZhouGHWB13}---an unsound but complete OWL 2 RL
approximation of the original ontology; thus, all answers that can be obtained
from the original ontology can also be obtained from the upper bound, but not
the other way around. Efficient materialisation of the upper bound was critical
for the work by \citeA{DBLP:conf/www/ZhouGHWB13}, and it has proved to be
challenging due to equality reasoning.

The left-hand part of Table \ref{tab:testStats} summarises our test data sets:
column `Rules' shows the total number of rules, column `sA-rules' shows the
number of rules containing the \sameAs property in the head, and column
`Triples before' shows the number of triples before materialisation.

\medskip

\noindent\textbf{Test Setting.}\quad
We conducted our tests on a Dell computer with 128~GB of RAM and two Xeon
E5-2643 processors with a total of 8 physical and 16 virtual cores, running
64-bit Fedora release 20, kernel version 3.13.3-201. We have not conducted warm
and cold start tests separately since, as a main-memory system, the performance
of RDFox should not be affected by the state of the operating system's buffers.
For the AX tests, we extended the relevant program with the seven rules from
Section \ref{sec:motivation}. In all cases we verified that the expansion of
the rewritten triples is identical to the triples derived using the
axiomatisation.

\medskip
\begin{table*}[tb]
\caption{Test Data Sets Before and After Materialisation}\label{tab:testStats}
\centering\scriptsize
\newcommand{\mr}[2][3]{\multirow{#1}{*}{#2}}
\def\mc[#1]#2{\multicolumn{#1}{c|}{#2}}
\vspace{2pt}
\begin{tabular}{|c|c|c|c|c|c|c|c|c|c|c|}
\hline
                & Rules         & sA-           & Triples   & Mode      & \mc[2]{Triples after} & Memory    & Rule      & Derivations   & Merged    \\
                                                                        \cline{6-7}
                &               & rules         & before    &           & unmarked  & total     & (GB)      & appl.     &               & resources \\
\hline
\hline
\mr[3]{Claros}  & \mr{1312}     & \mr{42}       & \mr{19M}  & AX        & \mc[2]{102M}          & 4.5       & 867M      & 11,009M       &           \\
                                                            \cline{5-11}
                &               &               &           & REW       & 79.5M     & 79.7M     & 3.6       & 149M      & 128M          & 12,890    \\
\cline{5-11}
                &               &               &           & factor    &           & 1.28x     & 1.28x     & 5.8x      & 85.5x         &           \\
\hline
\hline
\mr[3]{DBPedia} & \mr{3384}     & \mr{23}       & \mr{113M} & AX        & \mc[2]{139M}          & 6.9       & 934M      & 895M          &           \\
                                                            \cline{5-11}
                &               &               &           & REW       & 136M      & 136M      & 7.0       & 44.5M     & 37M           & 7,430     \\
                                                            \cline{5-11}
                &               &               &           & factor    &           & 1.2x      & 0.99x     & 21.0x     & 24.4x         &           \\
\hline
\hline
\mr[3]{OpenCyc} & \mr{261,067}  & \mr{3,781}    & \mr{2.4M} & AX        & \mc[2]{1,176M}        & 35.9      & 7,832M    & 12,890M       &           \\
                                                            \cline{5-11}
                &               &               &           & REW       & 141M      & 142M      & 4.6       & 309M      & 281M          & 361,386   \\
                                                            \cline{5-11}
                &               &               &           & factor    &           & 7.8x      & 7.8x      & 25.3x     & 45.9x         &           \\
\hline
\hline
\mr[3]{UniProt} & \mr{451}      & \mr{60}       & \mr{123M} & AX        & \mc[2]{228M}          & 15.1      & 1,801M    & 1,555M        &           \\
                                                            \cline{5-11}
                &               &               &           & REW       & 228M      & 228M      & 15.1      & 262M      & 183M          & 5         \\
                                                            \cline{5-11}
                &               &               &           & factor    &           & 1.0x      & 1.0x      & 6.9x      & 8.5x          &           \\
\hline
\hline
\mr[3]{UOBM}    & \mr{279}      & \mr{4}        & \mr{2.2M} & AX        & \mc[2]{36M}           & 1.2       & 332M      & 16,152M       &           \\
                                                            \cline{5-11}
                &               &               &           & REW       & 9.4M      & 9.7M      & 0.4       & 33.8M     & 4,256M        & 686       \\
                                                            \cline{5-11}
                &               &               &           & factor    &           & 3.2x      & 3.2x      & 9.9x      & 3.8x          &           \\
\hline
\end{tabular}
\end{table*}

\noindent\textbf{Effect of Rewriting on Total Work.}\quad
In order to see how rewriting affects the total amount of work, we materialised
each test data set in both AX and REW modes while collecting statistics about
the inference process; the results are shown in the right-hand part of Table
\ref{tab:testStats}. Column `Triples after' shows the number of triples after
materialisation; in the case of REW tests, we additionally show the number of
unmarked triples (i.e., of triples relevant to query answering). Column
`Memory' shows the total memory use as measured by RDFox's internal counters.
Column `Rule appl.' shows the total number of times a rule has been applied to
a triple, and column `Derivations' shows the total number of derivations.
Column `Merged resources' shows the number of resources that were replaced with
representatives in the course of materialisation. Finally, row `factor' shows
the ratio between the respective values in the AX and the REW tests.

As one can see, the reduction in the number of the derived triples is
correlated with the number of rewritten constants: on UniProt there is no
observable reduction since only five resources are merged; however, equalities
proliferate on OpenCyc and so rewriting is particularly effective. In all cases
the numbers of marked triples are negligible, suggesting that our decision to
mark, rather than delete triples does not have unexpected drawbacks. In
contrast, the reduction in the number of rule applications and, in particular,
of derivations is much more pronounced than the reduction in the number of
derived triples.

\medskip

\begin{table*}[tb]
\caption{Materialisation Times with Axiomatisation and Rewriting}\label{tab:testTimes}
\centering\scriptsize
\begin{tabular}{|c||r|r|r|r|r||r|r|r|r|r||r|r|r|r|r||}
\hline
Test&\multicolumn{5}{c||}{Claros} &\multicolumn{5}{|c||}{DBpedia}&\multicolumn{5}{|c||} {OpenCyc}\\
\hline
Threads
&\multicolumn{2}{|c|}{AX}&\multicolumn{2}{|c|}{REW}&\multirow{2}{*}{$\tfrac{\textrm{AX}}{\textrm{REW}}$}
&\multicolumn{2}{|c|}{AX}&\multicolumn{2}{|c|}{REW}&\multirow{2}{*}{$\tfrac{\textrm{AX}}{\textrm{REW}}$}
&\multicolumn{2}{|c|}{AX}&\multicolumn{2}{|c|}{REW}&\multirow{2}{*}{$\tfrac{\textrm{AX}}{\textrm{REW}}$}\\
\cline{2-5} \cline{7-10} \cline{12-15}
& sec & spd& sec&  spd&  & sec& spd& sec& spd&    &   sec& spd& sec & spd&     \\
\hline
1 & 2042.9 &1.0	&65.8& 1.0& 31.1 &219.8&1.0&31.7& 1.0& 6.9&2093.7& 1.0&119.9& 1.0&17.5\\
2 & 969.7  &2.1	&35.2& 1.9& 27.6 &114.6&1.9&17.6& 1.8& 6.5&1326.5& 1.6&	78.3& 1.5&16.9\\
4 & 462.0  &4.4	&18.1& 3.6& 25.5 &66.3&	3.3&10.7& 3.0& 6.2& 692.6& 3.0&	40.5& 3.0&17.1\\
8 & 237.2  &8.6	&9.9 & 6.7& 24.1 &36.1&	6.1& 5.2& 6.0& 6.9& 351.3& 6.0&	23.0& 5.2&15.2\\
12 & 184.9 &11.1&7.9 & 8.3& 23.3 &31.9&	6.9& 4.1& 7.7& 7.7& 291.8& 7.2&	56.2& 2.1& 5.5 \\
16 & 153.4 &13.3&6.9 & 9.6& 22.3 &27.5&	8.0& 3.6& 8.8& 7.7& 254.0& 8.2&	52.3& 2.3& 4.9\\
\hline
\multicolumn{10}{c}{\ }\\[-4pt]
\cline{1-11}
Test&\multicolumn{5}{|c||}{UniProt}&\multicolumn{5}{|c||}{UOBM}\\
\cline{1-11}
Threads
&\multicolumn{2}{|c|}{AX}&\multicolumn{2}{|c|}{REW}&\multirow{2}{*}{$\tfrac{\textrm{AX}}{\textrm{REW}}$}
&\multicolumn{2}{|c|}{AX}&\multicolumn{2}{|c|}{REW}&\multirow{2}{*}{$\tfrac{\textrm{AX}}{\textrm{REW}}$}\\
\cline{2-5} \cline{7-10}
&sec&spd& sec & spd&    & sec  & spd&   sec& spd&    \\
\cline{1-11}
1 &370.6& 1.0&143.4& 1.0& 2.6&2696.7& 1.0&1152.7& 1.0& 2.3\\
2 &232.3& 1.6& 86.7& 1.7& 2.7&1524.6& 1.8& 599.6& 1.9& 2.5\\
4 &129.2& 2.9& 46.5& 3.1& 2.8& 813.3& 3.3& 318.3& 3.6& 2.6\\
8 & 74.7& 5.0& 25.1& 5.7& 3.0& 439.9& 6.1& 177.7& 6.5& 2.5\\
12& 61.0& 6.1& 19.9& 7.2& 3.1& 348.9& 7.7& 152.7& 7.6& 2.3\\
16& 61.9& 6.0& 17.1& 8.4& 3.6& 314.4& 8.6& 137.9& 8.4& 2.3\\
\cline{1-11}
\end{tabular}
\end{table*}

\noindent\textbf{Effect of Rewriting on Materialisation Times.}\quad
In order to see how rewriting affects materialisation times, we measured the
wall-clock times needed to materialise our test data sets in AX and REW modes
on 1, 2, 4, 8, 12, and 16 threads. For each test, we report average wall-clock
time over three runs. Table \ref{tab:testTimes} shows our test results; column
`sec' shows the materialisation time in seconds, column `spd' shows the speedup
over the single-threaded version, and column
`$\tfrac{\textrm{AX}}{\textrm{REW}}$' shows the speedup of REW over AX.

As one can see from the table, RDFox parallelises computation exceptionally
well not only in AX mode but also in the REW mode which uses our extended algorithm. 
When using the eight physical cores of our test
server, the speedup is consistently between six and seven, which suggests that
the lock-free algorithms and data structures of RDFox are very effective. We
believe that the more-than-linear speedup on Claros is due to improved memory
locality resulting in fewer CPU cache misses. The speedup continues to increase
with hyperthreading, but is less pronounced: virtual cores do not provide
additional execution resources, and so they mainly compensate for CPU stalls
due to cache misses. The AX mode seems to scale better with the number of
threads than the REW mode, and we believe this to be due to contention between
threads while accessing the map $\rho$. Yet, the overall saved work compared to
the AX mode, makes more than up for it.
 Only OpenCyc in REW mode did not scale
particularly well: OpenCyc contains many rules, so sequentially updating $P$
and the associated rule index when $\rho$ changes becomes a significant
paralellisation bottleneck. Finally, since the materialisation of Claros with
more than eight threads in REW mode takes less than ten seconds, these results
are difficult to measure and are susceptible to skew.

Our results confirm that rewriting can significantly reduce materialisation
times. RDFox was consistently faster in the REW mode than in the AX mode even
on UniProt, where the reduction in the number of triples is negligible. This is
due to the reduction in the number of derivations, mainly involving rules
(\ref{eq:eq1})--(\ref{eq:eq5}), which is still significant on UniProt. In all
cases, the speedup of rewriting is typically much larger than the reduction in
the number of derived triples (cf.\ Table \ref{tab:testStats}), suggesting that
the primary benefit of rewriting lies in less work needed to match the rules,
rather than, as commonly thought thus far, in reducing the number of derived
triples. This is consistent with the fact that the speedup of rewriting was not
so pronounced on UniProt and UOBM, where the reduction in the number of
derivations was less significant.

Our analysis of the derivations that RDFox makes on UOBM revealed that, due to
the derived \sameAs triples, the materialisation contains large numbers of
resources connected by the \duri{hasSameHomeTownWith} property. This property
is also symmetric and transitive so, for each pair of connected resources, the
number of times each triple is derived by the transitivity rule is quadratic in
the number of connected resources. This leads to a large number of duplicate
derivations that do \emph{not} involve equality. Thus, although it is helpful,
rewriting does not reduce the number of derivation in the same way as, for
example, on Claros, which explains the relatively modest speedup of REW over AX.
\section{Conclusion}\label{sec:conclusion}

In this paper we have investigated issues related to the use of rewriting in
materialisation based OWL 2 RL systems. We have presented algorithms that
resolve these issues, and that can be effectively parallelised, and we have
shown empirically that our approach can reduce reasoning times on practical
data sets by orders of magnitude.

\clearpage



\bibliographystyle{aaai}

\ifdraft{
\clearpage
\appendix
\onecolumn

\section{Formalisation}\label{sec:reasoning:formalisation}

A rule $r$ was defined in Section~\ref{sec:preliminaries} as an implication of
the form \eqref{eq:rule-form}, where atom ${\head{r} \defeq \triple{s}{p}{o}}$
is the \emph{head} of $r$, conjunction ${\body{r} \defeq \triple{s_1}{p_1}{o_1}
\land\ldots\land \triple{s_n}{p_n}{o_n}}$ is the \emph{body} of $r$, and each
variable in $\head{r}$ also occurs in $\body{r}$. A \emph{program} $P$ is a
finite set rules. We also use the standard notions of a \emph{substitution}
$\sigma$ and composition $\sigma\tau$ of substitutions $\sigma$ and $\tau$; and
$\alpha\sigma$ is the result of applying $\sigma$ to a term, formula, or
program $\alpha$. Let $S$ be a finite set of facts. For $r$ a rule of the form
\eqref{eq:rule-form}, $r(S)$ is the smallest set such that ${H\sigma \in r(S)}$
for each substitution $\sigma$ satisfying ${B_i\sigma \in S}$ for each $i$ with
${1 \leq i \leq n}$; moreover, for $P$ a program, let ${P(S) \defeq \bigcup_{r
\in P} r(S)}$. Given a finite a set of \emph{explicit} (i.e., extensional or
EDB) facts $\efacts$, the \emph{materialisation} $P^\infty(\efacts)$ of $P$ on
$E$ is defined as follows: let ${P^0(\efacts) \defeq \efacts}$; let
${P^i(\efacts) \defeq P^{i-1}(\efacts) \cup P(P^{i-1}(\efacts))}$ for each ${i
> 0}$; and let ${P^\infty(\efacts) \defeq \bigcup_i P^i(\efacts)}$.

\subsection{Parallel Materialisation in RDFox}

For convenience, we will briefly recall some details of the RDFox algorithm
presented in \cite{MNPHO14a}. The RDFox algorithm computes $P^\infty(\efacts)$
using $N$ threads of a set of explicit facts $\efacts$ and a program $P$. Set
$\efacts$ is first copied into the set of \emph{all} facts $\tfacts$, after
which each thread starts updating $\tfacts$ using a fact-at-a-time version of
the semina{\"i}ve algorithm \cite{abiteboul95foundation}. In particular, a
thread selects an unprocessed fact $F$ from $\tfacts$ and tries to match it to
each body atom $B_i$ of each rule of the form \eqref{eq:rule-form} in $P$. For
each substitution $\sigma$ with ${F = B_i\sigma}$, the thread evaluates the
partially instantiated rule ${H\sigma \leftarrow B_1\sigma \land \dots \land
B_{i-1}\sigma \land B_{i+1}\sigma \land \dots \land B_k\sigma}$ by matching the
rule's body as a query in $\tfacts$, and adding $H\tau$ to $\tfacts$ for each
thus obtained substitution $\tau$ with $\sigma \subseteq \tau$. The thread
repeats these steps until all facts in $\tfacts$ have been processed.
Materialisation finishes if at this point all other threads are waiting;
otherwise, the thread waits for more facts to become available.

To implement this idea efficiently, RDFox stores all facts in $\tfacts$ in a
single table. As usual, resources are encoded using nonzero integer resource
IDs in a way that allows IDs to be used as array indexes. Furthermore, RDFox
maintains three array-based and three hash-based indexes that allow it to
efficiently identify all relevant facts in $\tfacts$ when matching a given
atom. Such a scheme has two important advantages. First, the indexes allow
queries (i.e., rule bodies) to be evaluated using nested index loop joins with
sideways information passing. Second, arrays and hash tables are naturally
parallel data structures and so they support efficient concurrent updates as
parallel threads derive fresh facts.

\subsection{Formalising the Rewriting Algorithm}

As discussed in Section~\ref{sec:reasoning}, we extend the RDFox algorithm with
rewriting to reduce both the size of $\tfacts$ and the time required to compute
it, while ensuring that an arbitrary SPARQL query can be answered exactly as if
it were evaluated over $[P \cup \Peq]^\infty(E)$.

We use \emph{short-circuit evaluation} of expressions: in `$A$ and $B$' (resp.\
`$A$ or $B$'), $B$ is evaluated only if $A$ evaluates to true (resp.\ false).
We store all facts in a data structure $\tfacts$ that must provide several
abstract operations: $\add{\tfacts}{F}$ atomically adds a fact $F$ to $\tfacts$
if $F$ is not already present in $\tfacts$ (marked or not), returning $\true$
if $\tfacts$ has been changed; and $\markout{\tfacts}{F}$ atomically marks a
fact ${F \in \tfacts}$ as outdated, returning $\true$ if $F$ has been changed.
Also, $\tfacts$ must provide an iterator over its facts: $\nnext{\tfacts}$
atomically selects and returns a fact or returns $\varepsilon$ if no such facts
exists; $\hasNext{\tfacts}$ returns $\true$ if $\tfacts$ contains such a fact;
and $\last{\tfacts}$ returns the last returned fact. These operations need not
enjoy the ACID properties, but they must be \emph{linearisable}
\cite{DBLP:books/daglib/0020056}: each asynchronous sequence of calls should
appear to happen in a sequential order, with the effect of each call taking
place at an instant between the call's invocation and response. Access to
$\tfacts$ thus does not require synchronisation via locks. Given a fact $F$
returned by $\nnext{\tfacts}$, let $\tfacts^{\prec F}$ be the facts returned by
$\nnext{\tfacts}$ before $F$, and let ${\tfacts^{\preceq F} \defeq
\tfacts^{\prec F}\cup\{ F \}}$.

For $\rho$ the mapping of resources to their representatives,
$\mergeInto{\rho}{d}{c}$ atomically checks whether $d$ is a representative of
itself; if so, it updates the representative of all resources that $d$
represents to the representative of $c$ and returns $\true$. We discuss how to
implement this operation and how to compute $\rho(\alpha)$ in the following
subsection. Moreover, $\rho(\tfacts)$ is the \emph{rewriting} of $\tfacts$ with
$\rho$, and ${\expand{\tfacts}{\rho} \defeq \{ \triple{s}{p}{o} \mid
\triple{\rho(s)}{\rho(p)}{\rho(o)} \in \tfacts \}}$ is the \emph{expansion} of
$\tfacts$ with $\rho$.

An \emph{annotated query} is a conjunction of atoms of the form
${A_1^{\bowtie_1} \land \dots \land A_k^{\bowtie_k}}$, where ${\bowtie_i \in \{
\prec, \preceq \}}$ for each ${1 \leq i \leq k}$. For $F$ a fact and $\sigma$ a
substitution, operation ${\tfacts.\evaluate(Q, F, \sigma)}$ returns the set
containing each minimal substitution $\tau$ such that ${\sigma \subseteq \tau}$
and $\tfacts^{\bowtie_i F}$ contains an unmarked fact $A_i\tau$ for each ${1
\leq i \leq k}$. Given a conjunction of atoms ${Q = B_1 \land \ldots \land
B_n}$, let ${Q^\preceq \defeq B_1^\preceq \land \ldots \land B_n^\preceq}$.

Finally, for $P'$ a set of rules and $F$ a fact, $P'.\rulesFor(F)$ returns each
tuple of the form ${\langle r, Q_i, \sigma\rangle}$ where ${r \in P'}$ is a
rule of the form \eqref{eq:rule-form}, $\sigma$ is a substitution such that ${F
= B_i\sigma}$, and ${Q_i = B_1^\prec\land\dots B_{i-1}^\prec\land
B_{i+1}^\preceq\land\dots\land B_k^\preceq}$.

To implement $\markout{\tfacts}{F}$, we associate with each fact a status bit,
which we update lock-free using compare-and-set operations
\cite{DBLP:books/daglib/0020056}; efficient implementation of all other
operations was described by \citeA{MNPHO14a}.

We use a queue $C$ of resources: $\enqueue{C}{c}$ atomically inserts a resource
$c$ into $C$; and $\dequeue{C}$ atomically selects and removes a resource from
$c$, or returns $\varepsilon$ if no such resource exists. We also use a queue
$R$ of rules. Lock-free implementation of these operations is described by
\citeA{DBLP:books/daglib/0020056}.
 
In addition to $\efacts$, $\tfacts$, $P$, $\rho$, $C$, and $R$, we use several
global variables: $N$ is the number of threads (constant); $W$ is the number of
waiting threads (initially 0); $P'$ is the `current' program (initially $P$);
$run$ is a Boolean flag determining whether materialisation should continue
(initially $\true$); $L$ is the last fact returned by $\nnext{\tfacts}$ before
$P'$ is updated (initially undefined); and $m$ is a mutex variable.

After initialising $\tfacts$ to $\efacts$, each of the $N$ threads executes
\refalg{mat}, trying in line \ref{alg:mat:if} to evaluate a rule whose
resources have been updated, rewrite facts containing an outdated resource, or
apply rules to a fact from $\tfacts$. When no work is available, the thread
enters a critical section (lines \ref{alg:mat:aM}--\ref{alg:mat:rM}) and waits
for more work or a termination signal (line \ref{alg:mat:more-work}). Variable
$W$ is incremented (line \ref{alg:mat:incW}) before entering, and decremented
(line \ref{alg:mat:decW}) after leaving the critical section, so at any point
in time $W$ is the number of threads inside the critical section. The thread
goes to sleep (line \ref{alg:mat:wait}) if no more work is available but other
threads are running. When the last thread runs out of work (line
\ref{alg:mat:last-enters}), it adds to $R$ an updated version of each outdated
rule in $P'$ (line \ref{alg:mat:update-rules}), notes the last fact in
$\tfacts$ (line \ref{alg:mat:note-last}), updates $P'$ (line
\ref{alg:mat:update-P}), signals termination if there are no rules to
reevaluate (line \ref{alg:mat:update-run}), and wakes up all waiting threads
(line \ref{alg:mat:notify}). Updating $P$ on a single thread simplifies the
implementation, but it introduces a potential sequential bottleneck; however,
our experiments have shown that, when $P$ is not too large, the amount of
sequential processing in lines \ref{alg:mat:update-rules}--\ref{alg:mat:notify}
does not significantly affect our approach.

\begin{algorithm}[t]
\caption{$\mathsf{materialise}$}\label{alg:mat}
\begin{algorithmic}[1]
    \item[\textbf{Global:}]
    \item[]
        \begin{tabular}{@{}r@{}l@{\,\;\;}l@{}l@{}}
            $N$               & : No. of threads                        & $W$       & : No. of waiting threads (0) \\
            $P$               & : a program                             & $P'$      & : the current program ($P$) \\
            $\efacts$         & : explicit facts                        & $\tfacts$ & : all facts ($\efacts$) \\
            $m$               & : a mutex variable                      & $L$       & : reevaluation limit (NaN) \\
            $\mathit{run}$    & : a Boolean flag ($\true$)              & $\rho$    & : resource mapping ($id$) \\
            $R$               & : a queue of rules ($\emptyset$)        & $C$       & : a queue of constants ($\emptyset$) \\[0.05in]
        \end{tabular}

    \While{$\mathit{run}$}
        \If{$\lnot \mathsf{evaluateUpdatedRules}()$ and $\lnot \mathsf{rewriteFacts}()$ and $\lnot \mathsf{applyRules}()$}      \label{alg:mat:if}
            \State increment $W$ atomically                                                                                     \label{alg:mat:incW}
            \MutexAcquire{\emph{m}}                                                                                             \label{alg:mat:aM}
            \While{$\isEmpty{R}\land\isEmpty{C}\land\lnot \hasNext{\tfacts}\land\mathit{run}$}                                  \label{alg:mat:more-work}
                \If{$W = N$}                                                                                                    \label{alg:mat:last-enters}
                    \State $R:=\{\rho(r)\mid r \in P' \text{ and } \rho(r) \not\in P' \}$                                              \label{alg:mat:update-rules}
                    \State $L \defeq \last{\tfacts}$                                                                            \label{alg:mat:note-last}
                    \State $P' \defeq \rho(P)$                                                                                  \label{alg:mat:update-P}
                    \State $\mathit{run} \defeq \isNotEmpty{R}$                                                                 \label{alg:mat:update-run}
                	\State \notifyThreads                                                                           \label{alg:mat:notify}
                \Else
                    \State \textbf{release} \emph{m}, wait for notification, \textbf{acquire} \emph{m}                          \label{alg:mat:wait}
                \EndIf
            \EndWhile
            \State decrement $W$ atomically                                                                                     \label{alg:mat:decW}
            \MutexRelease{\emph{m}}                                                                                             \label{alg:mat:rM}
        \EndIf
    \EndWhile
\end{algorithmic}
\end{algorithm}

\begin{figure*}[t] 
\begin{minipage}[t]{0.44\textwidth}
\begin{algorithm}[H]
\caption{$\mathsf{evaluateUpdatedRules}$}\label{alg:rule}
\begin{algorithmic}[1] 
    \State $r \defeq \dequeue{R}$
    \If{$r \neq \varepsilon$}
        \For{each $\tau \in \tfacts.\evaluate(\body{r}^\preceq,L,\emptyset)$}                                                   \label{alg:rule:for}
            \State \textbf{if} $\add{\tfacts}{\head{r}\tau}$ \textbf{then} \notifyThreads                           \label{alg:rule:add}
        \EndFor
    \EndIf
    \State \textbf{return} $r \neq \varepsilon$    
\end{algorithmic}
\end{algorithm}
\par
\begin{algorithm}[H]
\caption{$\mathsf{rewriteFacts}$}\label{alg:const}
\begin{algorithmic}[1]
    \State $c \defeq \dequeue{C}$
    \If{$c \neq \varepsilon$}
        \For{each unmarked fact $F \in \tfacts$ containing $c$}                                                              \label{alg:const:for}
			\State \textbf{if} $\markout{\tfacts}{F}$ and $\add{\tfacts}{\rho(F)}$ \textbf{then} \notifyThreads     \label{alg:const:rewrite}
        \EndFor
    \EndIf
    \State \textbf{return} $c \neq \varepsilon$                                                                                 \label{alg:const:return}
\end{algorithmic}
\end{algorithm}
\end{minipage}
\qquad
\begin{minipage}[t]{0.5\textwidth}
\begin{algorithm}[H]
\caption{$\mathsf{applyRules}$}\label{alg:fact}
\begin{algorithmic}[1]    
    \State $F \defeq \nnext{\tfacts}$                                                                                           \label{alg:fact:extract}
    \If{$F\neq \varepsilon$ and $F$ is not marked as outdated}
		\State $G \defeq \rho(F)$
        \If{$F \neq G$}                                                                                                         \label{alg:fact:if-outdated}
            \State \textbf{if} $\markout{\tfacts}{F}$ and $\add{\tfacts}{G}$ \textbf{then} \notifyThreads                       \label{alg:fact:rewrite}
        \ElsIf{$F$ is of the form $\triple{a}{\owl{sameAs}}{b}$}                                                                \label{alg:fact:sameas-start}
            \If{$a$ and $b$ are distinct}                                                                                       \label{alg:fact:ssameas}
                \State $c\defeq \min\{ a,b \}$; \quad $d \defeq \max\{ a,b \}$
                \If{$\mergeInto{\rho}{d}{c}$}
                    \State $\enqueue{C}{d}$ and \notifyThreads                                                                  \label{alg:fact:mergeconst}
                \EndIf
            \EndIf                                                                                                              \label{alg:fact:sameas-end}
        \ElsIf{$F$ is of the form $\triple{a}{\owl{differentFrom}}{a}$}                                                         \label{alg:fact:differentFrom}
            \State derive a contradiction and \notifyThreads                                                                    \label{alg:fact:contradiction}
        \Else
            \For{each $\langle r, Q, \sigma \rangle \in P'.\rulesFor(F)$}                                                       \label{alg:fact:ruleindex-start}
                \For{each $\tau \in \tfacts.\evaluate(Q, F, \sigma)$}
                    \State \textbf{if} $\add{\tfacts}{\head{r}\tau}$ \textbf{then} \notifyThreads                               \label{alg:fact:addinferred}
                \EndFor
			\EndFor                                                                                                             \label{alg:fact:ruleindex-end}
            \For{each resource $c$ occurring in $F$}                                                                            \label{alg:fact:ref-sameas-start}
            	\State \textbf{if} $\add{\tfacts}{\triple{c}{\owl{sameAs}}{c}}$ \textbf{then} \notifyThreads                    \label{alg:fact:ref-sameas-add}
            \EndFor                                                                                                             \label{alg:fact:ref-sameas-end}
        \EndIf
    \EndIf
    \State \textbf{return} $F \neq \varepsilon$
\end{algorithmic}
\end{algorithm}
\end{minipage}
\end{figure*}

Algorithm \ref{alg:rule} processes the updated rules in $R$ by evaluating their
bodies in $\tfacts^{\preceq L}$ and instantiating the rule heads. Algorithm
\ref{alg:const} rewrites all facts in $\tfacts$ that contain an outdated
resource $c$. Algorithm \ref{alg:fact} extracts from $\tfacts$ (line
\ref{alg:fact:extract}) an unprocessed, unmarked fact $F$ and processes it.
Fact $F$ is rewritten if it is outdated (lines
\ref{alg:fact:if-outdated}--\ref{alg:fact:rewrite}); this is needed because a
thread can derive a fact containing an outdated resource \emph{after} that
resource has been processed by Algorithm \ref{alg:const}. If $F$ is an \sameAs
triple with distinct resources (lines
\ref{alg:fact:sameas-start}--\ref{alg:fact:ssameas}), then the smaller resource
(according to an arbitrary total order) is selected as the representative for
the other one, and the latter is added to the queue $C$ of outdated resources
(line \ref{alg:fact:mergeconst}). An ordering on resources is needed to prevent
cyclic merges and to ensure uniqueness of the algorithm's result. The thread
derives a contradiction if $F$ is an \differentFrom triple with the same
resource (lines \ref{alg:fact:differentFrom}--\ref{alg:fact:contradiction}).
Otherwise, the thread applies the rules to $F$ (lines
\ref{alg:fact:ruleindex-start}--\ref{alg:fact:ruleindex-end}) and derives the
reflexive \sameAs triples (lines
\ref{alg:fact:ref-sameas-start}--\ref{alg:fact:ref-sameas-end}).

Theorem~\ref{thm:correctness} presented in Section \ref{sec:reasoning} captures
properties that ensure correctness of our algorithm; we restate the theorem and
present a detail proof in Appendix \ref{sec:appendix}.

\subsection{Implementing the Map of Representatives}\label{sec:reasoning:map}

Mapping $\rho$ consists of two arrays, $\msucc$ and $\mpred$, indexed by
resource IDs and initialised with zeros. Let $c$ be a resource. Then,
$\msucc[c]$ is zero if $c$ represents itself, or $\msucc[c]$ contains a
resource that $c$ has been merged into. To retrieve resources equal to a
representative, we organise each \sameAs-clique into a linked list of
resources, so $\mpred[c]$ contains the next pointer.

\refalg{merge} merges $d$ into $c$ in a lock-free way. We update $\msucc[d]$ to
$c$ if $d$ currently represents itself (line \ref{alg:merge:cas}). The
\emph{compare-and-set} primitive prevents thread interference:
$\cas{loc}{exp}{new}$ atomically loads the value stored at location $loc$ into
a temporary variable $old$, stores $new$ into $loc$ if ${old = exp}$, and
returns $old$. If this update is successful, we append the clique of $d$ to the
clique of $c$ (lines \ref{alg:merge:append-init}--\ref{alg:merge:append-end}):
we move to the end of $c$'s list (short-circuit evaluation ensures that CAS in
line \ref{alg:merge:append-start} is evaluated only if $\mpred[e] = 0$) and try
to change $\mpred[e]$ to $d$; if the latter fails due to concurrent updates, we
continue scanning $c$'s list.

\refalg{applyrho} computes $\rho(c)$ by traversing $\msucc$ until it reaches a
non-merged resource $r$. If another thread updates $\rho$ by modifying
$\msucc[r]$, we just continue scanning $\msucc$ past $r$, so the result is at
least as current as $\rho$ just before the start.

\begin{figure*}[t] 
\begin{minipage}[t]{0.70\textwidth}
\begin{algorithm}[H]
\caption{$\rho.\mathsf{mergeInto}(d,c)$}\label{alg:merge}
\begin{algorithmic}[1]
    \If{$\cas{\msucc[d]}{0}{c} = 0$}                                    \label{alg:merge:cas}
        \State $e \defeq c$                                             \label{alg:merge:append-init}
        \While{$\mpred[e] \neq 0$ or $\cas{\mpred[e]}{0}{d} \neq 0$}    \label{alg:merge:append-start}
            \State $e \defeq \mpred[e]$
        \EndWhile                                                       \label{alg:merge:append-end}
        \State \textbf{return} $\true$                                  \label{alg:merge:success}
    \Else
        \State \textbf{return} $\false$                                 \label{alg:merge:failure}
    \EndIf
\end{algorithmic}
\end{algorithm}
\end{minipage}
\qquad
\begin{minipage}[t]{0.25\textwidth}
\begin{algorithm}[H]
\caption{$\rho(c)$}\label{alg:applyrho}
\begin{algorithmic}[1]
    \State $r \defeq c$
    \Loop
        \State $r' \defeq \msucc[r]$
        \If{$r' = 0$}
            \State \textbf{return} $r$
        \Else
            \State $r \defeq r'$
        \EndIf
    \EndLoop
\end{algorithmic}
\end{algorithm}
\end{minipage}
\end{figure*}

\section{Proof of Theorem \ref{thm:correctness}}\label{sec:appendix}

\thmcorrectness*

For notational convenience, let ${\Pi^i \defeq [P \cup \Peq]^i(\efacts)}$ and
let ${\Pi^\infty \defeq [P \cup \Peq]^{\infty}(\efacts)}$. Furthermore, let
$N_r$ be the number of distinct resources occurring in $\efacts$, and let $|P|$
be the number of rules in $P$. We split our proof into several claims.

\begin{claim}
    The algorithm terminates.
\end{claim}

\begin{proof}    
Duplicate facts are eliminated eagerly, and facts are never deleted, so the
number of successful additions to $\tfacts$ is bounded by $N_r^3$. Moreover,
\refalg{merge} ensures that each resource is merged at most once; hence, $\rho$
can change at most $N_r$ times, and the number of additions to queue $C$ is
bounded by $N_r$ as well. Thus, $P' \neq \rho(P')$ may fail in lines
\reflines{mat}{last-enters}{notify} at most $N_r$ times, so the number of
additions to queue $R$ is bounded by ${|P| \cdot N_r}$. Together, these
observations clearly imply that the algorithm terminates.
\end{proof}

All operations used in our algorithm are linearisable and the algorithm
terminates, so the execution of $N$ threads on input $\efacts$ and $P$ has the
same effect as some finite sequence ${\Lambda = \langle \lambda_1, \dots,
\lambda_\ell \rangle}$ of operations where each $\lambda_i$ is
\begin{itemize}
    \item $\add{\tfacts}{F}$, representing successful addition of $F$ to
    $\tfacts$ in \refline{rule}{add}, \refline{const}{rewrite}, or line
    \ref{alg:fact:rewrite}, \ref{alg:fact:addinferred}, or
    \ref{alg:fact:ref-sameas-add} of \refalg{fact},
    
    \item ${F \defeq \nnext{\tfacts}}$, representing successful extraction of
    an unmarked, unprocessed fact $F$ from $\tfacts$ in \refline{fact}{extract},

    \item $\markout{\tfacts}{F}$, representing successful marking of $F$ as
    outdated in \refline{const}{rewrite} or \refline{fact}{rewrite},

    \item $\mergeInto{\rho}{d}{c}$, representing successful merging of resource
    $d$ into resource $c$ in \refline{fact}{mergeconst}, or
    
    \item $P' \defeq \rho(P)$, representing an update of program $P'$ in
    \refline{mat}{update-P}.
\end{itemize}
Materialisation first adds all facts in $\efacts$ to $\tfacts$, so the first
$m$ operations in $\Lambda$ are of the form $\add{\tfacts}{F_i}$ for each ${F_i
\in \efacts}$. By a slight abuse of notation, we often treat $\Lambda$ as a set
and write ${\lambda_i \in \Lambda}$. Our algorithm clearly ensures that each
operation $\add{\tfacts}{F}$ in $\Lambda$ is followed in $\Lambda$ by ${F
\defeq \nnext{\tfacts}}$ or $\markout{\tfacts}{F}$; if both of these operations
occur in $\Lambda$, then the former precedes the latter. Sequence $\Lambda$
induces a sequence of mappings of resources to representatives ${\rho_0,
\ldots, \rho_\ell}$: mapping $\rho_0$ is identity; for each $i > 0$ with
${\lambda_i = \mergeInto{\rho}{d}{c}}$, mapping $\rho_i$ is obtained from
$\rho_{i-1}$ by setting ${\rho_i(a) \defeq \rho_{i-1}(c)}$ for each resource
$a$ with ${\rho_{i-1}(a) = d}$; and for each $i > 0$ with ${\lambda_i \neq
\mergeInto{\rho}{d}{c}}$, let ${\rho_i \defeq \rho_{i-1}}$. Clearly, ${\rho =
\rho_\ell}$; furthermore, for each $i$ with ${1 \leq i \leq \ell}$, if
${\rho_i(F) \neq F}$, then ${\rho_j(F) \neq F}$ for each $j$ with ${i \leq j
\leq \ell}$.

\begin{claim}
    $\triple{a}{\sameAs}{b} \in \tfacts$ implies ${a = b}$.
\end{claim}

\begin{proof}
Assume that $\tfacts$ contains an unmarked fact $F$ of the form
$\triple{a}{\sameAs}{b}$ with ${a \neq b}$. Then, there exists an operation
${\lambda_i \in \Lambda}$ of the form ${F \defeq \nnext{\tfacts}}$. In case
${\rho_i(F) \neq F}$, then \reflines{fact}{if-outdated}{rewrite} ensure that
${\markout{\tfacts}{F} \in \Lambda}$, contradicting the assumption that $F$ was
unmarked. In case ${\rho_i(F) = F}$, then there exists an operation ${\lambda_j
\in \Lambda}$ with ${j \geq i}$ of the form $\mergeInto{\rho}{a}{b}$ or
$\mergeInto{\rho}{b}{a}$. Thus, either $a$ or $b$ is added to queue $C$ in
\refline{fact}{mergeconst}, and this resource is later processed in
\refalg{const}; then, due to \refline{const}{rewrite}, there exists an
operation ${\lambda_k \in \Lambda}$ with ${k \geq j}$ of the form
$\markout{\tfacts}{F}$. We obtain a contradiction in either case, as required.
\end{proof}

\begin{claim}
    ${F \in \tfacts}$ implies ${\rho(F) = F}$.
\end{claim}

\begin{proof}
Assume for the sake of contradiction that a fact ${F \in \tfacts}$ exists such
that ${F \neq \rho(F)}$. Then, there exists an operation ${\lambda_i \in
\Lambda}$ of the form ${F \defeq \nnext{\tfacts}}$. We clearly have $F =
\rho_i(F)$, or \reflines{fact}{if-outdated}{rewrite} would have ensured that
${\markout{\tfacts}{F} \in \Lambda}$. But then, there exists an operation
${\lambda_j \in \Lambda}$ with ${i < j \leq \ell}$ of the form
$\mergeInto{\rho}{d}{c}$ where resource $d$ occurs in $F$. Hence, $d$ is added
to queue $C$ in \refline{fact}{mergeconst}, and this resource is later
processed in \refalg{const}; but then, there exists an operation ${\lambda_k
\in \Lambda}$ with ${k \geq j}$ of the form $\markout{\tfacts}{F}$, which
contradicts our assumption that ${F \in \tfacts}$.
\end{proof}

\begin{claim}
    ${\tfacts^\rho \subseteq \Pi^\infty}$.
\end{claim}

\begin{proof}
The claim holds because $\Peq$ contains replacement rules
\eqref{eq:eq2}--\eqref{eq:eq4} and the following two properties are satisfied
for each ${1 \leq i \leq \ell}$:
\begin{enumerate}[(i)]
    \item for each resource $a$, we have ${\triple{a}{\sameAs}{\rho_i(a)} \in
    \Pi^\infty}$, and
    
    \item if $\lambda_i$ is of the form $\add{\tfacts}{F}$, then ${F \in
    \Pi^\infty}$.
\end{enumerate}
We prove (i) and (ii) by induction on $i$. For the induction base, we consider
the first $m$ operations in $\Lambda$ of the form $\add{\tfacts}{F_i}$ with
${F_i \in \efacts}$; both claims clearly hold for ${i = m}$. For the inductive
step, property (i) can be affected only if $\lambda_i$ is of the form
$\mergeInto{\rho}{d}{c}$, and property (ii) can be affected only if $\lambda_i$
is of the form $\add{\tfacts}{F}$; hence, we analyse these cases separately.

\medskip

Assume ${\lambda_i = \mergeInto{\rho}{d}{c}}$. To show that property (i) holds,
consider an arbitrary resource $a$; property (i) holds trivially for $a$ if
${\rho_{i-1}(a) \neq d}$, so assume that ${\rho_{i-1}(a) = d}$. Due to the form
of $\lambda_i$, constant $d$ was added to $C$ in \refline{fact}{mergeconst} due
to an operation ${\lambda_j \in \Lambda}$ with ${j < i}$ of the form ${F \defeq
\nnext{\tfacts}}$ with $F$ of the form $\triple{c}{\sameAs}{d}$ or
$\triple{d}{\sameAs}{c}$; but then, there exists an operation ${\lambda_k \in
\Lambda}$ with ${k < j}$ of the form $\add{\tfacts}{F}$; thus, by the induction
assumption, property (ii) implies ${F \in\Pi^\infty}$. Furthermore, by the
induction assumption and ${\rho_{i-1}(d) = d}$, we have
${\triple{a}{\sameAs}{d} \in \Pi^\infty}$, and we also have
${\triple{c}{\sameAs}{\rho_{i-1}(c)} \in \Pi^\infty}$. Moreover, ${\rho_i(a) =
\rho_{i-1}(c) = \rho_i(c)}$ holds by \refalg{merge}. Finally, property \sameAs
is reflexive, symmetric, and transitive in $\Pi^\infty$, so
${\triple{a}{\sameAs}{\rho_i(a)} \in \Pi^\infty}$ holds, as required for
property (i).

\medskip

Assume ${\lambda_i = \add{\tfacts}{F}}$ by \refline{const}{rewrite} or
\refline{fact}{rewrite}; thus, $F$ obtained from some fact $G$ for which there
exists an operation ${\lambda_j \in \Lambda}$ with $j < i$ of the form
$\add{\tfacts}{G}$. By the induction assumption, we have ${G \in \Pi^\infty}$.
Fact $G$ is obtained from $F$ by replacing each occurrence of a resource $c$ in
$F$ with $\rho_n(c)$ for some $n$ with ${1 \leq n \leq i}$; by the induction
assumption, mapping $\rho_n$ satisfies property (i), so we have
${\triple{c_}{\sameAs}{\rho_n(c)} \in \Pi^\infty}$. But then, replacement rules
\eqref{eq:eq2}--\eqref{eq:eq4} in $\Peq$ ensure that ${F \in \Pi^\infty}$.

\medskip

Assume ${\lambda_i = \add{\tfacts}{F}}$ by \refline{rule}{add} or
\refline{fact}{addinferred}; thus, $F$ is obtained by applying a rule
$\rho_i(r)$ of the form \eqref{eq:rule-form} via a substitution $\tau$ that
matches the body atoms ${B_1, \ldots, B_n}$ of $\rho_i(r)$ to facts ${F_1,
\ldots, F_n}$ such that, for each ${1 \leq j \leq n}$, sequence $\Lambda$
contains an operation of the form $\add{\tfacts}{F_j}$ preceding $\lambda_i$.
By the induction assumption, property (ii) implies ${\{ F_1, \ldots, F_n \}
\subseteq \Pi^\infty}$. We next show that rule $r$ can be applied to facts ${\{
F_1', \ldots, F_n' \} \subseteq \Pi^\infty}$ to derive a fact $F'$, and that
the rules in $\Peq$ can be used to derive $F$ from $F'$. Let ${C_r = \{ c_1,
\ldots, c_k \}}$ be the set of resources occurring in rule $r$; by the
induction assumption, property (i) ensures the following observation:
\begin{align}
    \triple{c_j}{\sameAs}{\rho_i(c_j)} \in \Pi^\infty \text{ for each } j \text{ with } 1 \leq j \leq k. \tag{$\lozenge$}
\end{align}
Let ${F_1', \ldots, F_n'}$ be the facts obtained from ${F_1, \ldots, F_n}$ by
replacing, for each ${c \in C_r}$, each occurrence of $\rho_i(c)$ with $c$; due
to ($\lozenge$), the rules in $\Peq$ ensure that ${\{F_1', \ldots, F_n' \}
\subseteq \Pi^\infty}$ holds. Now let $\sigma$ be the substitution obtained
from $\tau$ by replacing, for each ${c_j \in C_r}$, resource $\rho_i(c_j)$ in
the range of $\tau$ with $c_j$; then, substitution $\sigma$ matches all body
atoms of $r$ to derive fact ${F' \in \Pi^\infty}$ where ${\rho_i(F') = F}$. Due
to ($\lozenge$), the rules in $\Peq$ ensure ${F \in \Pi^\infty}$, as required
for property (ii).

\medskip

Assume ${\lambda_i = \add{\tfacts}{F}}$ by \refline{fact}{ref-sameas-add},
so ${F = \triple{c}{\sameAs}{c}}$ with $c$ a resource occurring in some fact
$G$ for which there exists an operation ${\lambda_j \in \Lambda}$ with ${j <
i}$ of the form $\add{\tfacts}{G}$. By the induction assumption, we have ${G
\in \Pi^\infty}$. But then, due to rules \eqref{eq:eq1} in $\Peq$, we have ${F
\in \Pi^\infty}$, as required.
\end{proof}

Before proving ${\tfacts^\rho \supseteq \Pi^\infty}$, we show a useful property
($\lozenge$) essentially saying that, whenever a fact $F$ is added to $\tfacts$
in operation $j$, at each step $i$ after $j$, a rewriting $G$ of $F$ is or will
be `visible' in $\tfacts$ (i.e., $G$ has not been marked outdated before
operation $i$).

\begin{claim}[$\lozenge$]
    For each $i$ with ${1 \leq i \leq \ell}$ and each operation ${\lambda_j \in
    \Lambda}$ with ${j \leq i}$ of the form $\add{\tfacts}{F}$, there exists
    $k$ such that
    \begin{enumerate}[(a)]
        \item ${\lambda_k = \add{\tfacts}{G}}$,
    
        \item $G$ is obtained from $F$ by replacing each occurrence of a resource
        $c$ with $\rho_n(c)$ for some $n$ with ${n \leq k}$, and
    
        \item for each $k'$ (if any) with ${k < k' \leq i}$, we have
        ${\lambda_{k'} \neq \markout{\tfacts}{G}}$.
    \end{enumerate}
\end{claim}

\begin{proof}
The proof proceeds by induction on $i$. The base case $i=0$ is vacuous, so we
assume that ($\lozenge$) holds up to some $i$ with ${1 \leq i < \ell}$, and we
consider operation $\lambda_{i+1}$ and an arbitrary operation $\lambda_j$ with
${j \leq i+1}$ of the form $\add{\tfacts}{F}$. If $j=i+1$, then the claim
trivially holds for ${k = j}$; otherwise, we have $j<i+1$, so by applying the
induction assumption to $i$, an integer $k$ with ${\lambda_k =
\add{\tfacts}{G}}$ satisfying properties (a)--(c). If ${\lambda_{i+1} \neq
\markout{\tfacts}{G}}$, then $k$ satisfies properties (a)--(c) for $i+1$ and
$\lambda_j$ as well. If, however, ${\lambda_{i+1} = \markout{\tfacts}{G}}$,
then either in \refline{const}{rewrite} or in \refline{fact}{rewrite} an
attempt will be made to add a fact $G'$ satisfying property (b) to $\tfacts$.
First, assume that $G'$ already exists in $\tfacts$---that is, some ${m \leq
i}$ exists such that ${\lambda_m = \add{\tfacts}{G'}}$; then, by applying the
induction assumption to $G'$, some $k'$ with ${\lambda_{k'} =
\add{\tfacts}{G''}}$ exists that satisfies properties (a)--(c); but then, $k'$
satisfies properties (a)--(c) for $F$ which proves the claim for $\lambda_j$.
In contrast, if no such $m$ exists, then the addition in
\refline{const}{rewrite} or \refline{fact}{rewrite} succeeds, and some $k'$
with ${i + 1 < k'}$ exists such that ${\lambda_{k'} = \add{\tfacts}{G'}}$.
Thus, properties (a)--(c) hold for $i+1$ and $\lambda_j$.
\end{proof}

\begin{claim}
    ${\tfacts^\rho \supseteq \Pi^\infty}$.
\end{claim}

\begin{proof}
The claim holds if ${\tfacts \supseteq \rho(\Pi^\infty)}$ and if
${\triple{c}{\sameAs}{d} \in \Pi^\infty}$ implies ${\rho(c) = \rho(d)}$. Thus,
we prove by induction on $i$ that each set $\Pi^i$ in the sequence ${\Pi^0,
\Pi^1, \ldots}$ satisfies the following two properties:
\begin{enumerate}[(i)]
    \item ${\tfacts \supseteq \rho(\Pi^i)}$ and

    \item ${\triple{c}{\sameAs}{d} \in \Pi^i}$ implies ${\rho(c) = \rho(d)}$.
\end{enumerate}
To prove these claims, we consider an arbitrary fact ${F \in \Pi^0}$ (for the
base case) or ${F \in \Pi^{i+1} \setminus \Pi^i}$ with ${i \geq 0}$ (for the
induction step) and show that ${\add{\tfacts}{\rho(F)} \in \Lambda}$. Since
$\rho$ is the final resource mapping, $F$ is never marked as outdated and so we
have ${\rho(F) \in \tfacts}$, as required for property (i). Moreover, if ${F =
\triple{c}{\sameAs}{d}}$, then ${\rho(F) \in \tfacts}$ together with property
(1) of Theorem \ref{thm:correctness} imply that $\rho(F)$ is of the form
$\triple{a}{\sameAs}{a}$, and so we have ${\rho(c) = a = \rho(d)}$, as required
for property (ii).

\medskip

\emph{Induction Base.}\; Consider an arbitrary fact ${F \in \Pi^0 = \efacts}$.
Let $G$ be the fact that satisfies (a)--(c) of property ($\lozenge$) for ${i =
\ell}$; fact $G$ is never marked as outdated due to (c), so ${\rho(F) = G}$
holds by property 2 of Theorem \ref{thm:correctness}. But then, property (a)
implies ${\add{\tfacts}{\rho(F)} \in \Lambda}$, as required.

\medskip

\emph{Induction Step.}\; Fact ${F \in \Pi^{i+1} \setminus \Pi^i}$ is derived
using a rule ${r \in P \cup \Peq}$ of the form \eqref{eq:rule-form} from facts
${\{ F_1, \ldots, F_n \} \subseteq \Pi^i}$. By the induction assumption we have
${\{ \rho(F_1), \ldots, \rho(F_n) \} \subseteq \tfacts}$, which implies
${\add{\tfacts}{\rho(F_j)} \in \Lambda}$ for each ${1 \leq j \leq n}$; we
denote the latter property with ($\dag$).

Assume that $F$ is derived by applying rule \eqref{eq:eq1} to a fact ${G \in
\Pi^i}$. By property ($\dag$), there exists an operation
${\add{\tfacts}{\rho(G)} \in \Lambda}$; thus, there exists an operation
${\rho(G) \defeq \nnext{\tfacts} \in \Lambda}$; finally,
\refline{fact}{ref-sameas-add} ensures ${\add{\tfacts}{\rho(F)} \in \Lambda}$.

Assume that $F$ is derived by applying rule \eqref{eq:eq2}, \eqref{eq:eq3}, or
\eqref{eq:eq4} to facts ${\{ G, \triple{c}{\sameAs}{d} \} \subseteq \Pi^i}$. By
the induction assumption, we have ${\rho(c) = \rho(d)}$. But then, since $G$ is
obtained from $F$ by replacing $c$ with $d$, we have ${\rho(G) = \rho(F)}$; by
property ($\dag$), we have ${\add{\tfacts}{\rho(G)} \in \Lambda}$, and so
${\add{\tfacts}{\rho(F)} \in \Lambda}$.

Assume that $F$ is derived by applying rule \eqref{eq:eq5} to a fact ${G \in
\Pi^i}$. By property ($\dag$), there exists an operation
${\add{\tfacts}{\rho(G)} \in \Lambda}$; thus, there exists an operation
${\rho(G) \defeq \nnext{\tfacts} \in \Lambda}$; finally,
\refline{fact}{contradiction} ensures ${\add{\tfacts}{\rho(F)} \in \Lambda}$.

Assume that $F$ is derived by applying a rule ${r \in P}$ of form
\eqref{eq:rule-form} to facts ${\{ F_1,\ldots, F_n \} \subseteq \Pi^i}$ via
some substitution $\sigma$. Let $\tau$ be the substitution where ${\tau(x) =
\rho(\sigma(x))}$ for each variable $x$ from the domain of $\sigma$. Rule
$\rho(r)$ derives $\rho(F)$ via $\tau$ from ${\rho(F_1), \ldots, \rho(F_n)}$.
To show that one can match these facts to an annotated query derived from
$\rho(r)$, let $G$ be the fact among ${\rho(F_1), \ldots, \rho(F_n)}$ for which
operation $\add{\tfacts}{G}$ occurs last in $\Lambda$, and let $j$ be the
smallest integer with ${1 \leq j \leq n}$ and ${\rho(F_j) = G}$. Rule $\rho(r)$
occurs in the final program $P'$, so we have two possibilities.
\begin{itemize}
    \item Assume that ${\rho(r) \not\in P}$ and that $\rho(r)$ occurs for the
    first time in $\Lambda$ in an operation ${P' \defeq \rho(P)}$ that appears
    in $\Lambda$ after operation $\add{\tfacts}{G}$. Then, by property
    ($\dag$), rule $\rho(r)$ is applied to facts ${\rho(F_1), \ldots,
    \rho(F_n)}$ in \refline{rule}{for}, and so ${\add{\tfacts}{\rho(F)} \in
    \Lambda}$ due to \refline{rule}{add}.
    
    \item In all other cases, $\Lambda$ contains an operation ${\rho(F_j)
    \defeq \nnext{\tfacts}}$, and at that point program $P'$ contains
    $\rho(r)$; hence, rule $\rho(r)$ is applied in
    \refline{fact}{ruleindex-start} by matching $\rho(B_j)$ to $G$. Now by
    ($\dag$) and the way in which we have selected $G$, we have ${\{ \rho(F_1),
    \ldots, \rho(F_n) \} \subseteq \tfacts^{\preceq G}}$, so each body atom
    $\rho(B_k)$ of $\rho(r)$ can be matched to $\tfacts^{\preceq G}$.
    Furthermore, $j$ is the smallest index such that ${\rho(B_j) = G}$, so
    ${\rho(F_k) \in \tfacts^{\prec G}}$ for each ${1 \leq k < j}$. Thus,
    substitution $\tau$ is returned in \refline{fact}{ruleindex-start}, and so
    \refline{fact}{addinferred} ensures that ${\add{\tfacts}{\rho(F)} \in
    \Lambda}$ holds, as required. \qedhere
\end{itemize}
\end{proof}

\begin{claim}
    Each pair of $r$ and $\tau$ is considered at most once either in
    \refline{rule}{for} or in \refline{fact}{ruleindex-start}.
\end{claim}

\begin{proof}
For the sake of contradiction, assume that a rule $r$ of the form
\eqref{eq:rule-form} and substitution $\tau$ exist that violate this claim. The
domain of $\tau$ contains all variables in $r$, so $\tau$ matches all body
atoms of $r$ to a unique set of facts ${F_1, \ldots, F_n}$. We next show that
the annotations in queries prevent the algorithm from considering the same $r$
and $\tau$ more than once. To this end, let ${G \in \{ F_1, \ldots, F_n \}}$ be
the fact for which operation $\add{\tfacts}{G}$ occurs last in $\Lambda$.

Assume that $\add{\tfacts}{G}$ occurs in $\Lambda$ before the operation ${P'
\defeq \rho(P)}$ in $\Lambda$ with ${r \in P'}$ but $r\notin P$. Since $P'$ is
updated in \refline{mat}{update-P} only when there are no facts to process,
operation ${G \defeq \nnext{\tfacts}}$ also occurs in $\Lambda$ before ${P'
\defeq \rho(P)}$; but then, $r$ cannot be applied to $G$ in
\refline{fact}{ruleindex-start}. Hence, the only possibility is that $r$ and
$\tau$ are considered twice is in \refline{rule}{for}; however,
\refline{mat}{update-rules} ensures that $r$ is enqueued into $R$ at most
once.

Assume that $\add{\tfacts}{G}$ occurs in $\Lambda$ after the operation ${P'
\defeq \rho(P) \in \Lambda}$ with ${r \in P'}$. \Refline{rule}{for} evaluates
the rules in $R$ only up to the last fact $L$ extracted from $\tfacts$ before
$P'$ is updated, and so $r$ cannot be matched in $G$ in \refline{rule}{for};
hence, the only possibility is that $r$ and $\tau$ are considered twice in
\refline{fact}{ruleindex-start}. To this end, assume that $F$ and $F'$ are (not
necessarily distinct) facts extracted in \refline{fact}{extract}, let $Q$ the
annotated query used to match body atom $B_i$ of $r$ to $F$, and let $Q'$ the
annotated query used to match body atom $B_j$ of $r$ to $F'$; thus, we have
${B_i\tau = F}$ and ${B_j\tau = F'}$. We consider the following two cases.
\begin{itemize}
    \item Assume ${F = F'}$; furthermore, w.l.o.g.\ assume that ${i \leq j}$.
    If ${i = j}$, we have a contradiction since operation ${F \defeq
    \nnext{\tfacts}}$ occurs in $\Lambda$ only once and query ${Q = Q'}$ is
    considered in \refline{fact}{ruleindex-start} only once. If ${i < j}$, we
    have a contradiction since ${\bowtie_i \, = \, <}$ holds in the annotated
    query $Q'$, so atom $B_i$ cannot be matched to fact $F$ in query $Q'$
    (i.e., we cannot have ${B_i\tau = F}$) due to ${F \not\in \tfacts^{< F'} =
    \tfacts^{< F}}$.

    \item Assume ${F \neq F'}$; furthermore, w.l.o.g.\ assume that operation
    ${\add{\tfacts}{F}}$ occurs in $\Lambda$ after operation
    ${\add{\tfacts}{F'}}$. But then, ${B_j\tau = F'}$ leads to a contradiction
    since atom $B_j$ cannot be matched to fact $F'$ in query $Q$ when fact $F$
    is extracted in \refline{fact}{extract}. \qedhere
\end{itemize}
\end{proof}
}{}

\end{document}